\documentclass[reqno,12pt]{amsart} 

\usepackage[T1]{fontenc}
\usepackage[utf8]{inputenc}
\usepackage[english]{babel}

\usepackage{amssymb}
\usepackage{amsmath}
\usepackage{tikz}
\usepackage{tikz-cd}
\usepackage{xcolor}
\newcommand{\braket}[2]{\left\langle #1 \vphantom{#2} \right.\left| \vphantom{#1} #2 \right\rangle }

\newcommand{\bra}[1]{\left\langle #1  \right|}
\newcommand{\ket}[1]{\left| #1 \right \rangle}

\newtheorem{thm}{Theorem}[section]

\newtheorem{theorem}[thm]{Theorem}
\newtheorem{corollary}[thm]{Corollary}

\newtheorem{remark}[thm]{\it Remark}

\begin{document}

\title{Discrete integrable  principal chiral field model  and its involutive reduction}


\author[J.L. Cieśliński]{Jan L. Cieśliński}
\address{J.L. Cieśliński, Wydział Fizyki, Uniwersytet w Białymstoku, ul. Ciołkowskiego 1L, 15-245 Białystok, Poland}
 \email{jancieslinski@gmail.com}
 
\author[A.V. Mikhailov]{Alexander V. Mikhailov}
\address{A.V. Mikhailov, Department of Applied Mathematics, University of Leeds, Leeds LS2 9JT, UK}
 \email{A.V.Mikhailov@leeds.ac.uk} 

\author[M. Nieszporski]{Maciej Nieszporski}
\address{Maciej Nieszporski, Wydział Fizyki,
Uniwersytet Warszawski, Pasteura 5, 02-093, Warsaw, Poland}
 \email{Maciej.Nieszporski@fuw.edu.pl}

\author[F.W. Nijhoff]{Frank W. Nijhoff}
\address{F.W. Nijhoff, Department of Applied Mathematics, University of Leeds, Leeds LS2 9JT, UK}
 \email{frank.nijhoff@gmail.com}


\begin{abstract}
We discuss 
an integrable discretization of the principal chiral field models equations   and its involutive reduction.
We present a Darboux transformation  and general construction of solution solutions for these discrete equations.
\end{abstract}

\maketitle

{\em Dedicated to the memory of David Kaup.}

\section{Introduction}

Given a differential equation of mathematical physics that is integrable one can search for a difference (discrete)
equation such that, first, in a continuum limit goes to the differential one, second, is integrable as well.
This article can be viewed as the first step in a larger program of integrable discretizations of chiral field model equations (or  $\sigma$-models equations). Therefore before we summarize the results of this article
we would like to  glance at the problem from a broad perspective and to present our motivations first. 
Namely, by the non-autonomous (or non-isospectral) extension of principal $GL(N)$ chiral field model equations  we understand the equation
\begin{equation} \label{ContinuousR}
(r  \Phi,_v \Phi^{-1}),_u+(r \Phi,_u \Phi^{-1})),_v=0, \qquad  r,_{uv} =0,
\end{equation}
where  $\Phi$ is a function of two independent  $u$ and $v$ variables  (in the hyperbolic case $u$ and $v$ are real variables while in the elliptic case $u$ and $v$ are treated as complex variables such that one variable is complex conjugation of the second)
that takes values in invertible square matrices of given size $N$ with complex entries, while $r$ is a scalar (in general complex valued) function of the independent variables. 

The chiral models (or nonlinear $\sigma$ models) appeared in theoretical physics  in sixties of the previous century in two classical articles of particle physics \cite{Gur,Gel}.
The history of these models in mathematics is even longer for in the beginning of the previous century \cite{Bia} Bianchi constructed B\"acklund transformation for the system
\begin{equation}
\label{Bianchi}
{\mathbf N},_{uv}=f {\mathbf N}, \quad {\mathbf N} \cdot {\mathbf N}=U(u)+V(v), 
\end{equation}
where ${\mathbf N}$ is a vector valued  function of variables $u$ and $v$, dot denotes (pseudo)-scalar product and $U(u)$ and $V(v)$ are prescribed functions of indicated variable. Nowadays  Bianchi system~(\ref{Bianchi}) is referred to as $O(p,q)$ $\sigma$-model ($(p,q)$ is the signature of the dot product). In the case $p=3$, $q=0$ the Bianchi system is equivalent to an involutive reduction of 
 $U(2)$ version of model (\ref{ContinuousR})
and (in the case $p=2$, $q=1$) to hyperbolic Ernst equation of general relativity.
Indeed, it turns out the following two constraints are compatible with
equations~(\ref{ContinuousR}) 
\begin{enumerate}
\item  involutive constraint
\begin{equation}
\begin{array}{c}
\Phi^2={\mathbb I}
\end{array}
\end{equation}
which gives after the substitution $\Phi={\mathbb I}-2 P$  idempotent constraint on $P$ i.e. $P^2=P$.
\item  (pseudo)-unitary constraint 
\begin{equation}
\label{u}
\begin{array}{cccccc}
\Phi^{\dagger} \eta \Phi \eta ={\mathbb I} &
~~~~~&\eta:=diag &(  
 \underbrace{1,...,1} ,&  
 \underbrace{-1,...,-1})   
 \\
&&&i&j
\end{array}
\end{equation}
\end{enumerate}
In the case $N=2$ the constraints give 
\[ \Phi=\frac{1}{\sqrt{r}}\begin{bmatrix}
   n_1 &  n_2+i n_3 \cr \epsilon (n_2-i n_3)  & -n_1
  \end{bmatrix},   n_1^2+\epsilon (n_2^2+n_3^2)=r, \qquad\epsilon=\pm 1
\]
 and  equation (\ref{ContinuousR}) becomes equation  (\ref{Bianchi}) for ${\mathbf N}:= [n_1,n_2,n_3]$,
where ${\mathbf N} \cdot {\mathbf N}:=n_1^2+\epsilon (n_2^2+n_3^2)=r$.
After the stereographical projection $n_2+i n_3=\sqrt{r} \frac{2 \xi}{1+\epsilon \bar{\xi} \xi}$, 
 $n_1=\sqrt{r} \frac{1-\epsilon \bar{\xi} \xi}{1+\epsilon \bar{\xi} \xi}$ we come to the equation
\begin{equation}
\label{e}
\begin{array}{c}
(\bar{\xi} \xi+\epsilon) 
(\xi,_{uv}+\frac{r,_v}{2r} \xi,_u+\frac{r,_u}{2r} \xi,_v)= 2 \bar{\xi}  \xi,_u \xi,_v
\\
r,_{uv}=0
\end{array}
\end{equation} 
which is  (in the elliptic case and $\epsilon=-1$) the Ernst equation i.e. the equation
the problem 
of finding vacuum solutions of Einstein equations outside  axisymmetric and stationary
distribution of matter can be reduced to \cite{Ern}.
A bit more tedious calculations shows that in the case $N=3$ and $U(1,2)$ and involutive reduction leads to
the Ernst Einsten-Maxwell equations \cite{ErnEM}
\begin{equation}
\label{em}
\begin{array}{c}
(\bar{\xi} \xi- \bar{\chi} \chi-1) 
(\xi,_{uv}+\frac{r,_v}{2r} \xi,_u+\frac{r,_u}{2r} \xi,_v)= 
\xi,_u (\bar{\xi} \xi,_v-\bar{\chi}\chi,_v)+\xi,_v(\bar{\xi} \xi,_u-\bar{\chi}\chi,_u)
\\
(\bar{\xi} \xi- \bar{\chi} \chi-1) 
(\chi,_{uv}+\frac{r,_v}{2r} \chi,_u+\frac{r,_u}{2r} \chi,_v)=
\chi,_u(\bar{\xi}\xi,_v-\bar{\chi}\chi,_v)+\chi,_v (\bar{\xi}\xi,_u-\bar{\chi}\chi,_u)
\\
r,_{uv}=0.
\end{array}
\end{equation}
 Electro-magnetic field outside the distribution of matter
manifests in this description as an additional complex variable $\chi$.

The essential fact from the point of view of this article is that the differential equations presented above are integrable.
There is a large number of publications on the topic. We confine ourselves to listing 
the most important ones
\begin{itemize}
\item In 1905 \cite{Bia} Bianchi published B\"acklund transformation for the system (\ref{Bianchi})
\item
In 1972  
Geroch found an infite-dimensional group (the Geroch group $K$) equation (\ref{e}) is covariant under \cite{Ger}. Later on Kinnersley extended Geroch's result to electrovacuum equations
(\ref{em}) finding a covariance group (the Kinnersley group $K'$) for equations~(\ref{em})~\cite{Kin}.
\item
In 1978 the Ernst equations became
a part of `soliton' theory \cite{Mas,BeZa1,Har,Neu,OMW,Ale1}.
\item 
In 1978 Zakharov and Mikhailov  established the general scheme for integration chiral fields and their reductions \cite{MiZa,Jan1,Sasha}.
\end{itemize}
There are partial results in 
discretization of system (\ref{ContinuousR}) and its reductions, 
 we are going to recall them now.
\begin{itemize}
\item To the best of our knowledge, an integrable discretization of autonomous (or isospectral) ($r=const$) version of Bianchi system
(\ref{Bianchi}) was given by Orfanidis \cite{Orf}  and reads
\begin{equation}
\label{iso-d}
T_1T_2\vec{n}+ \vec{n} = F (T_1\vec{n}+ T_2 \vec{n}) \qquad \vec{n}\cdot \vec{n}=1
\end{equation}
Where  this time the dependent variables are $\vec{n}:   {\mathbb Z}^2  \mapsto T\mathbb{E}^n$ (where $T\mathbb{E}^n$ denotes tangent space of (pseudo)-Euclidean  space), 
$F:   {\mathbb Z}^2  \mapsto \mathbb{R}$, 
is a scalar function that can be expresses in terms of $\vec{n}$ as follows $F=\frac{2 \vec{n}\cdot (T_1\vec{n}+ T_2 \vec{n}) }{(T_1\vec{n}+ T_2 \vec{n}) \cdot (T_1\vec{n}+ T_2 \vec{n}) }$.
Equation (\ref{iso-d}) can be written in the case $n=3$ as
\begin{equation}
\label{s2}
\Delta_1 \left(\Delta_2 (\Phi) \Phi\right)+\Delta_2 \left(\Delta_1 (\Phi) \Phi \right)=0, \quad \Phi^{\dagger} \Phi={\mathbb I}, \quad \Phi^2={\mathbb I}
\end{equation}
where $\Phi$ is 2x2 matrix such that 
\[\Phi:=\begin{bmatrix}n_0 & n_1+in_2\cr n_1-in_2 & -n_0\end{bmatrix}, \qquad \vec{n}:=(n_0,n_1,n_2).\]
\item 
As for discretization of the non-isospectral (non-autonomous) case only two results are known so far.
The first one is the system introduced by Schief \cite{Sch}
\begin{equation}
\label{sch}
\begin{array}{c}
T_1T_2\vec{N} + \vec{N} = F  (T_1\vec{N} + T_2\vec{N} ),\\
\Delta_1 \Delta_2 [(T_1T_2\vec{N} + \vec{N} )\cdot (T_1\vec{N} + T_2\vec{N} )]= 0,
\end{array}
\end{equation}
as a permutability theorem of Calapso equation. It has been demonstrated that Schief's
system can be regarded as discretization of Bianchi system (\ref{Bianchi}) \cite{NDS,DNS,DNS1}
but not the only one. The second result is that the system (\ref{sch}) is a potential version
of the system \cite{DNS}
\begin{equation}
\label{dns}
\begin{array}{c}
\frac{T_1T_2\vec{\mathcal N} }{T_1T_2\sqrt{F} }+ \frac{\vec{\mathcal N} }{\sqrt{F}} = T_1\sqrt{F}  T_1\vec{\mathcal N} + T_2\sqrt{F} T_2\vec{\mathcal N}, \\
\Delta_1 \Delta_2 (\vec{\mathcal N} \cdot \vec{\mathcal N})= 0.
\end{array}
\end{equation}
The solution spaces (\ref{sch}) and (\ref{dns}) are related by
\[\vec{\mathcal N}=\sqrt{F}(T_1\vec{N} + T_2\vec{N} ).\]
\item Autonomous discrete principal GL(n) $\sigma$-models has been introduced
by Cherednik~\cite{Che} and read
\begin{equation}
\Delta_1 \left(\Delta_2 (\Phi) \Phi^{-1}\right)+\Delta_2 \left(\Delta_1 (\Phi) \Phi^{-1}\right)=0.
\end{equation}
A Lax pair for the equation and discussion on some integrable  features can be found in \cite{NiSm,Bog}.
\end{itemize}
In the present article we add to the list the system
\begin{equation}
\label{chep2C}
\Delta_1 \left(\Delta_2 (\Phi) \Phi\right)+\Delta_2 \left(\Delta_1 (\Phi)  \Phi\right)=0, \quad \Phi^2={\mathbb I},
\end{equation}
which is a discrete integrable version of  the autonomous ($r=1$), hyperbolic version of equation (\ref{ContinuousR}), subjected to involutive costraint.
We also discuss equation~(\ref{s2})  (i.e. system~(\ref{chep2C}) subjected to unitary constraint) in the case of 2x2 matrices,
and show  that system  matrices~(\ref{chep2C})  of larger size than 2, in contrast to the continuous case, do not admit unitary reduction.

\section{The classical principal chiral model and involutive reduction}

We study one possible integrable discretization of the classical principal chiral field model:
\begin{equation} \label{Continuous}
( \Phi ,_u \Phi^{-1}),_v + ( \Phi,_v \Phi^{-1}), _u = 0, \qquad \Phi \in GL(N, \mathbb{C}).
\end{equation}
where $\Phi$ is $GL(N, \mathbb{C})$-valued function of two independent real variables $u$ and $v$.

The integrability of the principal chiral field model was established in \cite{MiZa}, where equation~\eqref{Continuous} was represented as the compatibility condition of two linear problems:
\begin{equation} \label{LPCi}
\Psi,_u= \frac{{\mathcal A}}{\lambda-1} \Psi, \quad \Psi,_v= \frac{\mathcal B }{\lambda+1} \Psi.
\end{equation}
Here  $\lambda$ is a spectral parameter.

The compatibility conditions of (\ref{LPCi}) read
\[{\mathcal A},_v-{\mathcal B} ,_u=0, \quad {\mathcal A},_v+{\mathcal B },_u+{\mathcal A} {\mathcal B } -{\mathcal B } {\mathcal A}=0,\]
and guarantee the existence of a common fundamental solution $\Psi(u,v;\lambda)$.
From the second equation we infer that there exists potential $\Phi$ such that
\[{\mathcal A}=-\Phi,_u\Phi^{-1}, \qquad {\mathcal B }=\Phi,_v\Phi^{-1}.\]
Substituting these expressions into the first equation yields equation \eqref{Continuous}. The potential $\Phi$ can be identified with $\Psi(u,v;0)$.

Equation \eqref{Continuous} can naturally be reduced to any complex or real classical Lie group \cite{MiZa_cmp}, assuming that the variables ${\mathcal A}, {\mathcal B }$ are elements of the corresponding Lie algebra. In \cite{MiZa}, it was shown that the principal chiral field model \eqref{Continuous} admits a remarkable (see the introduction) integrable reduction:
\begin{equation}
\label{i}
\Phi^2 =  {{\mathbb I}}.
\end{equation}
We refer to this as the involutive reduction or projective (idempotent) reduction, since condition \eqref{i} implies the representation $\Phi = {\mathbb I} - 2P$, where $P$ is a projector ($P^2 = P$). In terms of $P$, equation \eqref{Continuous} takes the elegant form \cite{MiZa}:
\[
 [P_{u v},P]=0.
\]

In this paper we study an integrable discretization
\begin{equation}
\label{che}
\Delta_1 ((\Delta_2 \Phi)\Phi^{-1})+\Delta_2 ((\Delta_1 \Phi) \Phi^{-1})=0,
\end{equation}
of the principle chiral model \eqref{Continuous},
where $\Phi$ is $GL(N, \mathbb{C})$-valued function of two discrete variables $(m_1,m_2)\in\mathbb{Z}^2$
and $\Delta_i$ denotes forward difference operator in variable $m_i$, i.e. $\Delta_1 f(m_1,m_2):=f(m_1+1,m_2)-f(m_1,m_2)$, 
$\Delta_1 f(m_1,m_2):=f(m_1,m_2+1)-f(m_1,m_2)$, for any function $f:\mathbb{Z}^2\mapsto\mathbb{C}$. In section \ref{sec2} we construct explicitly two poles  Darboux-B\"acklund transformations for equation \eqref{Continuous}.
In section \ref{sec3} we show that the involutive reduction $\Phi^2={\mathbb I}$ is compatible with the discrete system \eqref{Continuous}
 (a constraint is compatible with an equation  iff its  imposition on initial conditions
implies that due to the equation the constraint holds at every point of the domain), leading to an additional symmetry of the Lax representation and the corresponding fundamental solution. In section \ref{sec31} we construct a two poles Darboux matrix and the corresponding soliton solution of the reduced   system
\begin{equation}
\label{chep2}
\Delta_1 ((\Delta_2 \Phi) \Phi )+\Delta_2 ((\Delta_1 \Phi) \Phi )=0,\quad \Phi^2={\mathbb I},
\end{equation}
where $\Phi\in GL(N, \mathbb{C})$. This solution is parameterized by two points on the complex Grassmannian $G_{k,N}$ and two complex parameters $\lambda_0$ and $\mu_0$ (Corollary 2). In the  case $N=2$ solutions for the system (\ref{chep2}) are presented in Section \ref{sec32}. This particular case ($N=2$)
admits the unitary reduction $\Phi^{-1}=\Phi^{\dagger}=\Phi$ and solutions for the system
(\ref{chep2}) subjected  unitary constraint $\Phi^{\dagger} \Phi ={\mathbb I} $ are presented in Section \ref{unitary}.
In this case the solution obtained is parameterized by one point on the $CP^1$ and one complex 
parameter $\lambda_0$.

\section{The $GL(n)$ principal $\sigma$-models in the discrete case}
\label{sec2}
We start our considerations from the pair of linear equations (the so called Lax pair) 
on   function $\Psi (m_1,m_2;\lambda)$  \cite{NiSm}
(two independent variables are integers and are omitted to make formulae shorter)
\begin{equation}
\label{Lp}
T_1 \Psi (\lambda) = \left( {\mathbb I} + \frac{1}{\lambda-1}{\mathbb A} \right) \Psi (\lambda)\qquad 
T_2 \Psi (\lambda) = \left( {\mathbb I} + \frac{1}{\lambda+1}{\mathbb B} \right)\Psi(\lambda)
\end{equation}
where $\lambda$ is a  complex valued parameter referred to as spectral parameter and in general 
can depend on $m_1$ and $m_2$,
however, in the present paper we confine ourselves 
to the isospectral case, i.e. to the case when  $\lambda$ does not depend on the independent variables.
We assume that the square matrices ${\mathbb A}$ and ${\mathbb B}$ do not depend on $\lambda$ and do depend on $m_1$ and $m_2$,
 ${\mathbb I}$ denotes the unit matrix
and we recall that  $T_j$ ($j=1,2$) denotes forward shift operators: $T_1\Psi (m_1,m_2):=\Psi (m_1+1,m_2)$, $T_2\Psi (m_1,m_2):=\Psi (m_1,m_2+1)$.
The function  $\Psi $
is a  fundamental matrix solution of the system (\ref{Lp}) so it takes values in square matrices.
Shortly, 
 we have \[T_1\Psi (\lambda)= U (\lambda) \Psi(\lambda),  \qquad T_2 \Psi(\lambda) = V (\lambda) \Psi(\lambda),\] 
where $U$ and $V$ are characterized as having simple poles at $\lambda = 1$, and $\lambda=-1$, respectively, and
\begin{equation} \label{limy}
\lim_{\lambda \rightarrow \infty} U (\lambda) = {\mathbb I} \ , \qquad 
\lim_{\lambda \rightarrow \infty} V (\lambda) = {\mathbb I} \ .
\end{equation} 
We want the compatibility conditions of system (\ref{Lp})
\begin{equation}
(T_2 {\mathbb A} -{\mathbb A}-T_1 {\mathbb B}  +{\mathbb B}) \lambda +(T_2 {\mathbb A}-{\mathbb I}) ({\mathbb B}+{\mathbb I}) -(T_1 B+{\mathbb I})  ({\mathbb A}-{\mathbb I})=0.
\end{equation}
be valid for all $\lambda \in \bar{\mathbb C}\backslash \{1,-1\}$, so we arrive at
\begin{equation}
\label{a}
\Delta_2 {\mathbb A} = \Delta_1 {\mathbb B} 
\end{equation}
\begin{equation}
\label{m}
[T_2({\mathbb A}-{\mathbb I})] ({\mathbb B}+{\mathbb I}) =[T_1({\mathbb B}+{\mathbb I})]  ({\mathbb A}-{\mathbb I})
\end{equation}
Due to equation (\ref{m}) we can introduce\footnote{Due to (\ref{a}), one can introduce `potential' $\Gamma$ via ${\mathbb A}=\Delta_1 \Gamma+{\mathbb I}, \, {\mathbb B}=\Delta_2 \Gamma-{\mathbb I}$ and, plugging it in (\ref{m}), get 
$T_2(\Delta_1 \Gamma)\Delta_2\Gamma=T_1(\Delta_2\Gamma)\Delta_1\Gamma$. Alternatively one can put
${\mathbb A}=2\Delta_1 H, \, {\mathbb B}=\Delta_2 H$ and get $\Delta_1 \Delta_2 H=T_2(\Delta_1 H)\Delta_2 H-T_1(\Delta_2 H)\Delta_1 H$
} 
`potential' $\Phi$
\begin{equation}
\label{P}
 {\mathbb A}-{\mathbb I}= - (T_1\Phi)  \Phi ^{-1}   \qquad {\mathbb B}+{\mathbb I}=  (T_2 \Phi) \Phi ^{-1}.
\end{equation}
Thus
\begin{equation}  \label{ala}
{\mathbb A} = - ( \Delta_1 \Phi ) \Phi^{-1} \ , \qquad {\mathbb B} = ( \Delta_2 \Phi ) \Phi^{-1} \ .
\end{equation} 
The potential $\Phi$ can be defined as
\begin{equation}
\label{M}
\Phi=\Psi(0). 
\end{equation}
Substitution (\ref{ala}) to equation (\ref{a}) gives
\begin{equation}
\label{sigma1}
T_1T_2\Phi (T_1\Phi^{-1}+T_2\Phi^{-1})=(T_1\Phi+T_2\Phi)\Phi^{-1}
\end{equation}
which is nothing but equation (\ref{che}).
To prove the integrability of (\ref{sigma1}) we will construct its Darboux-B\"acklund transformation. 

\subsection{Construction of a Darboux matrix}
We consider the Darboux transformation 
\begin{equation} \label{DT}
\tilde \Psi(\lambda) = {\mathbb D} (\lambda) \Psi(\lambda)
\end{equation} 
which leads to the Darboux-B\"acklund transformation
\begin{equation}  \label{DBT}
\tilde U(\lambda) = (T_1 {\mathbb D} (\lambda)) U(\lambda) {\mathbb D} (\lambda)^{-1} \ , \quad \tilde V(\lambda) = (T_2 {\mathbb D} (\lambda)) V (\lambda){\mathbb D} (\lambda)^{-1} \ ,
\end{equation} 
where ${\mathbb D} $ is usually refered to as Darboux matrix and depends on $m_1, m_2$ and $\lambda$. 

The normalization ${\mathbb N}$ of the Darboux matrix  is defined by 
\begin{equation} \label{normalization}
{\mathbb N} := \lim_{\lambda \rightarrow \infty} {\mathbb D} (\lambda) \ .
\end{equation} 
we tacitly assume that the limit exists.
First, we are going  
to check whether condition~(\ref{limy}) is preserved by the transformation. 
Taking into account (\ref{limy}) and (\ref{normalization}) we have
\begin{equation}  \begin{array}{l} 
\lim_{\lambda \rightarrow \infty} \tilde U (\lambda) = (T_1 {\mathbb N}) (\lim_{\lambda \rightarrow \infty}  U (\lambda) ) {\mathbb N}^{-1} = ( T_1 {\mathbb N} ) {\mathbb N}^{-1} \ ,  \\[2ex]
\lim_{\lambda \rightarrow \infty} \tilde V (\lambda) = (T_2 {\mathbb N}) (\lim_{\lambda \rightarrow \infty}  V (\lambda) ) {\mathbb N}^{-1} = ( T_2 {\mathbb N} ) {\mathbb N}^{-1} \ . 
\end{array} \end{equation} 
These limits are equal to ${\mathbb I}$ if and only if $T_1 {\mathbb N} = T_2 {\mathbb N} = {\mathbb N}$.
\begin{corollary}
If ${\mathbb N} = {\rm const}$, then the Darboux transformation (\ref{DBT}) preserves the constraints~(\ref{limy}).
\end{corollary}
We assume throughout the paper ${\mathbb N} = {\mathbb I}$.

There is another, more fundamental, constraint to be preserved by the transformation~(\ref{DBT}), namely, the dependence of  $U, V$  
on $\lambda$ (the divisors of poles of $U$ and $V$). 
We adapt the classical results on $N$-pole Darboux matrix 
 for continuous $GL(n)$ principal $\sigma$-model to the discrete $GL(n)$ principal $\sigma$-model
(\ref{che}), 
i.e.
we confine ourselves to the following form of ${\mathbb D} $ and its inverse ${\mathbb D}^{-1}$
\begin{equation}  \label{multiD}
 {\mathbb D} = {\mathbb I} + \sum_{i=1}^N \frac{\lambda_i P_i}{\lambda - \lambda_i} \ , \qquad {\mathbb D}^{-1} = {\mathbb I} + \sum_{j=1}^N \frac{\mu_j F_j}{\lambda - \mu_j} \ ,\qquad  \forall i,j \, : \, \lambda_i\ne\mu_j, \quad \lambda_i, \, \mu_i \neq 1
\end{equation} 
where the rank of all matrices $P_j$, $F_j$ is the same and equals $k$, i.e., 
\begin{equation}  \label{ketAB}
P_j = \ket{p_{1j}} \bra{p_{2j}} \ , \qquad F_j = \ket{f_{1j}} \bra{f_{2j}} \ ,
\end{equation} 
where ``bras'' are $k\times n$ matrices of rank $k$
while ``kets'' are $n \times k$ matrices of rank $k$.

The conditions ${\mathbb D} {\mathbb D}^{-1} = {\mathbb D}^{-1} {\mathbb D} = {\mathbb I}$ impose the following constraints:
\begin{equation}  \begin{array}{l}  \label{wiezy1}
 \bra{p_{2i}} {\mathbb D}^{-1} (\lambda_i) = 0 \ , \qquad {\mathbb D} (\mu_i) \ket{f_{1i}} = 0 \ , \\[2ex]
{\mathbb D}^{-1} (\lambda_i) \ket{p_{1i}} = 0 \ , \qquad \bra{f_{2i}} {\mathbb D} (\mu_i) = 0 \ ,
\end{array} \end{equation} 
which allow us to express $\ket{p_{1i}}$ and $\bra{f_{2i}}$ in terms of $\ket{f_{1i}}$ and $\bra{p_{2i}}$.

To ensure  that the transformation (\ref{DBT}) preserves divisors of poles of $U$ and $V$,
 we demand that the residua at $\lambda_i$, $\mu_i$ ($i=1,\ldots,N$) 
of the right hand sides of (\ref{DBT}) vanish:
\begin{equation}  \begin{array}{l}  \label{DBT-eqs} 
T_1 (P_i) U (\lambda_i) {\mathbb D}^{-1} (\lambda_i) = 0 \ , \\[2ex]
T_2 (P_i) V (\lambda_i) {\mathbb D}^{-1} (\lambda_i) = 0 \ , \\[2ex]
T_1 ({\mathbb D} (\mu_i) ) U (\mu_i) F_i = 0 \ , \\[2ex]
T_2 ({\mathbb D} (\mu_i) ) V (\mu_i) F_i = 0 \ .
\end{array} \end{equation}

The following theorem takes place in full analogy with the continuous case \cite{MiZa,Jan2,Jan3}:
\begin{theorem}  \label{Darb}
If $\bra{p_{2i}} = \bra{p_{0i}} \Psi^{-1} (\lambda_i) $ and $\ket{f_{1i}} = \Psi (\mu_i)\ket{f_{0i}} $, where $\Psi (\lambda)$ 
is a solution of the linear problem (\ref{Lp}) 
and $\bra{p_{0i}}$ and $\ket{f_{0i}} $ are respectivelly constant
$k\times n$ matrices of rank $k$
and $n \times k$ matrices of rank $k$ and $\ket{p_{1i}}$ and $\bra{f_{2i}}$ are given via formulae (\ref{wiezy1}),
then the equations
(\ref{DBT-eqs}) are satisfied.
\end{theorem}
\begin{proof}
The proof is straightforward. It is enough to substitute
$\bra{p_{2i}} = \bra{p_{0i}} \Psi^{-1} (\lambda_i) $ and $\ket{f_{1i}} = \Psi (\mu_i)\ket{f_{0i}} $
to (\ref{DBT-eqs})
and 
to take into account $T_1 \Psi (\lambda) = U (\lambda) \Psi (\lambda)$, $T_2 \Psi (\lambda) = V(\lambda) \Psi (\lambda)$ evaluated at $\lambda=\lambda_i$ 
or $\lambda=\mu_i$. The resulting expression vanishes by virtue of (\ref{wiezy1}). E.g. for the first equation of 
(\ref{DBT-eqs}) we have
\[
T_1 (P_i) U (\lambda_i) {\mathbb D}^{-1} (\lambda_i) = (T_1 \ket{p_{1i}}) (T_1 \bra{p_{2i}}) U (\lambda_i) {\mathbb D}^{-1} (\lambda_i)=\]
\[
=(T_1 \ket{p_{1i}}) \bra{p_{0i}} (T_1 \Psi^{-1} (\lambda_i)) U (\lambda_i) {\mathbb D}^{-1} (\lambda_i)=(T_1 \ket{p_{1i}}) \bra{p_{0i}}  \Psi^{-1} (\lambda_i) {\mathbb D}^{-1} (\lambda_i)=
\]
\[=
(T_1 \ket{p_{1i}}) \bra{p_{2i}} {\mathbb D}^{-1} (\lambda_i)=0.\]
\end{proof}

\begin{remark}
The factors  $ \ket{p_{1i}}$ , $\bra{p_{2i}}$,   $ \ket{f_{1i}}$ and
$\bra{f_{2i}}$ in (\ref{ketAB}) are not uniquely defined. Indeed, a transformation
\[
\ket{p_{1i}}\mapsto\ket{p_{1i}}W_i^{-1}, \  \bra{p_{2i}}\mapsto W_i\bra{p_{2i}},\quad  \ket{f_{1i}}\mapsto\ket{f_{1i}}\hat W_i^{-1}, \  \bra{f_{2i}}\mapsto \hat W_i\bra{f_{2i}},                                                                                \]
where $W_i,\hat W_i$ are any invertible $k\times k$ matrices, does not change the matrices $P_i,\ F_i$. Now, it follows from (\ref{wiezy1}) and Theorem \ref{Darb}) that the dressing Darboux matrix (\ref{multiD}) is parameterized by $2N$ points  $ \bra{p_{0i}}, \ \ket{f_{0i}} $ on the Grassmanian $G_{k,N}$ and by the set of $2N$ complex parameters $\lambda_i,\mu_i,\ \ i=1,\ldots,N$.
\end{remark}
\begin{remark}
 In (\ref{ketAB}) we assumed that all poles of the Darboux matrices in the spectral parameter $\lambda$ are simple, that the number of poles in ${\mathbb D}$ and ${\mathbb D}^{-1}$ coincide, and that the ranks of matrices $P_i,F_i$ are the same and do not depend on $i$. A more general Darboux dressing matrix can be easily constructed without the above assumptions and in the same way but it would make the paper less readable.
\end{remark}

\subsection{Darboux matrix with two poles} 
\label{forth}
Since in the case of reduction $\Phi^2 = {\mathbb I}$ considered in next section it is sufficient to consider
Darboux matrix with two poles we discuss the case $N=2$ in detail.
In the case $N=2$ Darboux transformation, Darboux matrix and its inverse take respectively forms
\begin{equation} \label{Dt2}
\tilde{\Psi}(\lambda;\lambda_1,\lambda_2)= {\mathbb D} (\lambda;\lambda_1,\lambda_2)\Psi(\lambda).
\end{equation} 
\begin{equation}
\label{D-two}
{\mathbb D} (\lambda) := {\mathbb I} +\frac{\lambda_1 {\mathbb P}}{\lambda-\lambda_1}+
\frac{\lambda_2 {\mathbb Q}}{\lambda - \lambda_2},
\end{equation}
\begin{equation}
\label{Dinv-two}
{\mathbb D}^{-1}(\lambda) := {\mathbb I} +\frac{\mu_1 {\mathbb F}}{\lambda-\mu_1}+
\frac{\mu_2 {\mathbb G}}{\lambda - \mu_2}.
\end{equation}
Taking into account ${\mathbb D} {\mathbb D}^{-1} = {\mathbb D}^{-1} {\mathbb D} = {\mathbb I}$, we arrive at the equations
\begin{equation}
\label{PQGF}
\begin{array}{l}
{\mathbb P} ({\mathbb I}+ \frac{\mu_1}{\lambda_1-\mu_1} {\mathbb F}+ \frac{\mu_2}{\lambda_1 - \mu_2} {\mathbb G})=0\\
{\mathbb Q} ({\mathbb I}+ \frac{\mu_1}{\lambda_2 - \mu_1} {\mathbb F}+ \frac{\mu_2}{\lambda_2 - \mu_2} {\mathbb G})=0\\
({\mathbb I}+ \frac{\lambda_1}{\mu_1-\lambda_1} {\mathbb P} + \frac{\lambda_2}{\mu_1 - \lambda_2} {\mathbb Q}){\mathbb F} =0\\
({\mathbb I}+ \frac{\lambda_1}{\mu_2 - \lambda_1} {\mathbb P} + \frac{\lambda_2}{\mu_2 - \lambda_2} {\mathbb Q}){\mathbb G} =0.\\
\end{array}
\end{equation}
We assume that  all matrices ${\mathbb P}, {\mathbb Q}, {\mathbb F}, {\mathbb G}$ are of the same rank $k>0$ 
so can be written as 
\[{\mathbb P}=\ket{p_1} \bra{p_2},\, {\mathbb Q}=\ket{q_1} \bra{q_2},\, {\mathbb F}=\ket{f_1} \bra{f_2},\,{\mathbb G}=\ket{g_1} \bra{g_2},\]
where ``bras'' are $k\times n$ matrices of rank $k$
while ``kets'' are $n \times k$ matrices of rank $k$.  
Equations~(\ref{PQGF}) take form
\begin{equation}
\label{pqgf1}
\begin{array}{l}
  \bra{p_2} + \frac{\mu_1}{\lambda_1-\mu_1}  \braket{p_2}{f_1} { \bra{f_2}}+ \frac{\mu_2}{\lambda_1 - \mu_2} \braket{p_2}{g_1} 
{ \bra{g_2}} =0\\
  \bra{q_2} + \frac{\mu_1}{\lambda_2 - \mu_1}
\braket{q_2}{f_1} { \bra{f_2}}+ \frac{\mu_2}{\lambda_2 - \mu_2}  \braket{q_2}{g_1} { \bra{g_2}} =0 \\
\ket{f_1} + \frac{\lambda_1}{\mu_1-\lambda_1} { \ket{p_1}} \braket{p_2}{f_1} + \frac{\lambda_2}{\mu_1 - \lambda_2} 
{ \ket{q_1}}  \braket{q_2}{f_1}   =0\\
\ket{g_1} + \frac{\lambda_1}{\mu_2 - \lambda_1} { \ket{p_1}} \braket{p_2}{g_1} + \frac{\lambda_2}{\mu_2 - \lambda_2} { \ket{q_1}} 
\braket{q_2}{g_1}  =0, \\
\end{array}
\end{equation}
where $\braket{p_2}{f_1}$, $\braket{p_2}{g_1}$, $\braket{q_2}{f_1}$ and $\braket{p_2}{f_1}$ are matrices $k\times k$.
The system (\ref{pqgf1})
can be solved  with respect to  $\ket{p_1}$, $\ket{q_1}$, $\bra{f_2}$ and $\bra{g_2}$,
therefore
 $\ket{p_1}$, $\ket{q_1}$, $\bra{f_2}$ and $\bra{g_2}$ are given in terms of $\bra{p_2}$, $\bra{q_2}$, $\ket{f_1}$ and $\ket{g_1}$.
The latter set, according to theorem~\ref{Darb}, evolves in a simple way and can be integrated, namely
$\bra{p_2}=\bra{p_{20}} \psi (\lambda_1)$, $\bra{q_2}=\bra{q_{20}}\psi (\lambda_2)$, 
$\ket{f_1}=\ket{f_{10}} \psi (\mu_1)$ and $\ket{g_1}=\ket{g_{10}}\psi (\mu_2)$
where $\bra{p_{20}}$, $\bra{q_{20}}$, $\ket{f_{10}}$ and $\ket{g_{10}}$ are constant bras and kets.
Upon the observation that in equation (\ref{pqgf1}) $\bra{p_2}$ and $\bra{q_2}$ are combinations of $\ket{f_2}$, $\ket{g_2}$  only,
 while  
$\ket{f_1}$ and $\ket{g_1}$ are  combinations of $\bra{p_1}$ and $\bra{q_1}$ only, 
one can obtain formulas for  $\ket{p_1}$, $\ket{q_1}$, $\bra{f_1}$, $\bra{g_1}$ by inverting $2\times 2$ block matrices.
We will need here only formulas for  $\ket{p_1}$, $\ket{q_1}$ which are
{\tiny
\begin{equation} 
\label{greek1}
\begin{array}{l}
\ket{p_1}= \frac{1}{\lambda_1}
\left((\mu_1-\lambda_2)\ket{f_1} \braket{q_2}{f_1}^{-1}-(\mu_2-\lambda_2)\ket{g_1} \braket{q_2}{g_1}^{-1} \right) 
\left( \frac{\mu_2-\lambda_2}{\mu_2-\lambda_1}  \braket{p_2}{g_1} \braket{q_2}{g_1}^{-1}-
       \frac{\mu_1-\lambda_2}{\mu_1-\lambda_1}  \braket{p_2}{f_1} \braket{q_2}{f_1}^{-1}          \right)^{-1} ,\\
\ket{q_1}= \frac{1}{\lambda_2}
\left((\mu_1-\lambda_1)\ket{f_1} \braket{p_2}{f_1}^{-1}-(\mu_2-\lambda_1)\ket{g_1} \braket{p_2}{g_1}^{-1} \right) 
\left( \frac{\mu_2-\lambda_1}{\mu_2-\lambda_2}  \braket{q_2}{g_1} \braket{p_2}{g_1}^{-1}-
       \frac{\mu_1-\lambda_1}{\mu_1-\lambda_2}  \braket{q_2}{f_1} \braket{p_2}{f_1}^{-1}\right)^{-1} ,
\end{array}
\end{equation} }
 formulas for $\bra{f_2}$ and $\bra{g_2}$ will not be used and we omit them.

\section{Involutive reduction $\Phi^2={\mathbb I}$}
\label{sec3}
As far as the reduction $\Phi^2={\mathbb I}$ is concerned we are guided by the results obtained in the continuous case (see \cite{MiZa}). 
The first observation is that
the constraint $\Phi^2={\mathbb I}$ is valid reduction of (\ref{che}) for it is preserved under propagation of $\Phi$ by means of equation (\ref{che}). Indeed,
if at some point of the lattice $\Phi^2={\mathbb I}$ and at neighboring points $T_1\Phi^2={\mathbb I}$ and $T_2\Phi^2={\mathbb I}$ then due to (\ref{che})
we have
\begin{equation} 
\begin{array}{l}
 T_1T_2\Phi^2= 
 \left( T_1\Phi +T_2 \Phi \right) \Phi^{-1}\left( T_1\Phi^{-1} +T_2 \Phi^{-1} \right)^{-1}\left( T_1\Phi +T_2 \Phi \right) \Phi^{-1}\left( T_1\Phi^{-1} +T_2 \Phi^{-1} \right)^{-1}= \\
 \left( T_1\Phi +T_2 \Phi \right) \Phi\left( T_1\Phi +T_2 \Phi \right)^{-1}\left( T_1\Phi +T_2 \Phi \right) \Phi \left( T_1\Phi +T_2
 \Phi \right)^{-1}= 
 {\mathbb I}
 \end{array}
\end{equation} 
The second observation is  that 
if the constraint
\begin{equation}
\label{con-p2}
\Phi^2= {\mathbb I}  
\end{equation}
holds then the matrices
\[ U(\lambda):=\frac{\lambda {\mathbb I} - (T_1\Phi)\Phi ^{-1}}{\lambda -1} ,\qquad 
   V(\lambda):=\frac{\lambda {\mathbb I} + (T_2\Phi)\Phi ^{-1}}{\lambda +1}\] 
   of the Lax pair
 \begin{equation} \label{Lpl}
 T_1\Psi(\lambda) = U (\lambda) \Psi(\lambda), \qquad T_2 \Psi(\lambda) = V (\lambda) \Psi(\lambda) 
\end{equation}  
has the property
\begin{equation} \label{property}
U(\lambda^{-1})=(T_1\Phi)U(\lambda)\Phi ^{-1},\qquad 
  V(\lambda^{-1})=(T_2\Phi)V(\lambda)\Phi ^{-1}.
\end{equation}

Due to  the property (\ref{property}),  matrix $\Phi\Psi(\lambda^{-1})$ is a fundamental solution of the Lax pair
(\ref{Lpl})
thus we have
\begin{equation}
\label{con-p}
\Psi  (\lambda^{-1})= \Phi \Psi (\lambda) {\mathbb S} \, 
\end{equation}
valid for every $\lambda \in \bar{\mathbb C}$ and where ${\mathbb S}$
is an invertible constant matrix. In general matrix ${\mathbb S}$ can depend on $\lambda$,
but we shall see that assumption that ${\mathbb S}$ does not on depends on $\lambda$ leads to nontrivial Darboux transformation.
We will usually choose such a fundamental solution of $\Psi (\lambda)$ such that  $\Psi (0)=\Phi$.
In the later case evaluating (\ref{con-p}) at $\lambda=0$ and $\lambda=\infty$
we get 
\begin{equation}
\label{norm}
\Psi (\infty)= {\mathbb S}, \quad  \Psi (\infty) {\mathbb S} = {\mathbb I},
\end{equation}
i.e.
\[{\mathbb S}^2={\mathbb I}.\]
After above preliminary considerations we are ready to construct Darboux transformation preserving constraint $\Phi^2 ={\mathbb I}$.

\subsection{The Darboux-B\"acklund transformation preserving constraint $\Phi^2 ={\mathbb I}$}\label{sec31}
Now we demand the constraints (\ref{con-p}),
to be preserved under the Darboux-B\"acklund transformation 
\begin{equation} \label{Dtp}
\tilde{\Psi }(\lambda)={\mathbb D} (\lambda) \Psi  (\lambda)  \ ,
\end{equation} 
i.e. 
\begin{equation} \label{con-pt}
\tilde{\Psi }(\lambda^{-1})=\tilde{\Phi} \tilde{\Psi } (\lambda) \tilde{\mathbb S} .
\end{equation} 
Due to chosen normalization of the Darboux matrix $D(\infty)={\mathbb I}$ we get from (\ref{Dtp}) $\tilde{\Psi }(\infty)=\Psi  (\infty)$
and due to (\ref{norm})  $\tilde{\mathbb S} ={\mathbb S}$ holds.
We immediately (by evaluating (\ref{con-pt}) at $\lambda=\infty$) get that $\tilde{\Psi } (0)=\tilde{\Phi} $
and we arrive at the following constraint on the Darboux matrix ${\mathbb D} $: 
\begin{equation}  \label{D-reduction1}
{\mathbb D} (\lambda^{-1}) \Phi = {\mathbb D} (0) \Phi {\mathbb D} (\lambda).
\end{equation}
Conversely if $D(\lambda)$ obey the constraint  (\ref{D-reduction1}) then matrix $\tilde{\Psi } (\lambda)$ given by (\ref{Dtp})
satisfies (\ref{con-pt}) where $ \tilde{\Phi}= \tilde{\Psi }(0) $ and as a consequence
\begin{equation}
\label{z2}
({\mathbb D} (0) \Phi)^2 = {\mathbb I} \ 
\end{equation}
i.e.
\begin{equation}
\tilde{\Phi}^2={\mathbb I}
\end{equation}

From (\ref{D-reduction1}) it follows that the set of poles of ${\mathbb D} (\lambda)$ coincides with the set of poles 
of ${\mathbb D} (1/\lambda)$. In the case of the two-pole Darboux matrix (\ref{D-two}) it means that either $\lambda_1 \lambda_2 = 1$ 
or $\lambda_1^2 = \lambda_2^2 = 1$. We confine ourselves to the first (generic) case.
The same considerations applies to the inverse 
of ${\mathbb D} $,  see (\ref{Dinv-two}). We also confine ourselves to the case $\mu_1 \mu_2 = 1$. 
Therefore we assume function~${\mathbb D} (\lambda)$ has two poles disposed symmetrically with respect to inversions in the unit sphere  
\begin{equation} \label{lao}
\lambda_1 = \lambda_0, \qquad \lambda_2 = \lambda_0^{-1}, \qquad \lambda_0 \neq 0,1
\end{equation}
i.e. 
\begin{equation}
\label{Dlm0}
{\mathbb D} (\lambda) := {\mathbb I} +\frac{\lambda_0 {\mathbb P}}{\lambda-\lambda_0}+
\frac{{\mathbb Q}}{\lambda \lambda_0-1},
\end{equation}
and its inverse is of the same form ($\mu_1 = \mu_0$, $\mu_2 = 1/\mu_0$, $\mu_0 \neq 0$), i.e., 
\begin{equation}
\label{Dlmi0}
{\mathbb D}^{-1}(\lambda) := {\mathbb I} +\frac{\mu_0 {\mathbb F}}{\lambda-\mu_0}+
\frac{{\mathbb G}}{\lambda \mu_0-1}.
\end{equation}
i.e.
\begin{equation} \label{muo}
\mu_1 = \mu_0, \qquad \mu_2 = \mu_0^{-1}, \qquad \mu_0 \neq 0,1.
\end{equation}

There are two simple conditions sufficient to satisfy (\ref{D-reduction1}): 
\begin{theorem}  \label{tw3}
If 
\begin{equation} \label{pq}
\bra{q_2} = \bra{p_2} \Phi , \qquad \ket{g_1} = \Phi \ket{f_1}, 
\end{equation}
then  constraint (\ref{D-reduction1}) is satisfied. 
\end{theorem}

{\it Proof:} 
Indeed, condition (\ref{D-reduction1}) ${\mathbb D} (\lambda^{-1}) \Phi {\mathbb D} (\lambda)^{-1}= {\mathbb D} (0) \Phi $ after substitution of assumed forms
of ${\mathbb D}$ (\ref{Dlm0}) and ${\mathbb D}^{-1}$ (\ref{Dlmi0}) reads
\begin{equation} \label{main}
\left({\mathbb I} +\frac{\lambda \lambda_0 {\mathbb P}}{1-\lambda \lambda_0}+
\frac{\lambda {\mathbb Q}}{ \lambda_0-\lambda}\right) \Phi \left({\mathbb I} +\frac{\mu_0 {\mathbb F}}{\lambda-\mu_0}+
\frac{{\mathbb G}}{\lambda \mu_0-1}\right)=\left({\mathbb I}-{\mathbb P}-{\mathbb Q}\right)\Phi
\end{equation} 
It is clear that residual at $\infty$ of left hand side and right hand side of (\ref{main}) are the same.
The remaining residue of left hand side at $\lambda_0$, $\lambda_0^{-1}$, $\mu_0$, $\mu_0^{-1}$ vanish iff
respectively
\begin{equation} \begin{array}{l} \label{lammu}
\ket{q_1}\bra{p_2}\left({\mathbb I} +\frac{\mu_0}{\lambda_0-\mu_0} {\mathbb F}+ \frac{1}{\lambda_0 \mu_0-1}{\mathbb G}\right)=0, \\ 
\ket{p_1}\bra{q_2}\left({\mathbb I} +\frac{\lambda_0 \mu_0}{1-\lambda_0\mu_0} {\mathbb F}+ \frac{\lambda_0}{\mu_0-\lambda_0}{\mathbb G}\right)=0, \\
\left({\mathbb I} +\frac{\lambda_0 {\mathbb P}}{\mu_0- \lambda_0}+ \frac{1}{ \lambda_0\mu_0 -1} {\mathbb Q} \right) \ket{f_1}\bra{g_2} =0\\
\left({\mathbb I} +\frac{\lambda_0 \mu_0 {\mathbb P}}{1-\lambda_0 \mu_0}+ \frac{\mu_0 }{ \lambda_0-\mu_0}{\mathbb Q}\right) \ket{g_1}\bra{f_2} =0
\end{array} \end{equation} 
hold, where we have used conditions (\ref{pq}).
The point is equations (\ref{lammu}) are satisfied due to the fact that we already made sure that equations (\ref{PQGF}) hold and equations
(\ref{PQGF}) in virtue~of~(\ref{lao}) and (\ref{muo}) take form
\begin{equation} \begin{array}{l} \label{mammu}
\ket{p_1}\bra{p_2}\left({\mathbb I} +\frac{\mu_0}{\lambda_0-\mu_0} {\mathbb F}+ \frac{1}{\lambda_0 \mu_0-1}{\mathbb G}\right)=0, \\ 
\ket{q_1}\bra{q_2}\left({\mathbb I} +\frac{\lambda_0 \mu_0}{1-\lambda_0\mu_0} {\mathbb F}+ \frac{\lambda_0}{\mu_0-\lambda_0}{\mathbb G}\right)=0, \\
\left({\mathbb I} +\frac{\lambda_0 {\mathbb P}}{\mu_0- \lambda_0}+ \frac{1}{ \lambda_0\mu_0 -1} {\mathbb Q} \right) \ket{f_1}\bra{f_2} =0\\
\left({\mathbb I} +\frac{\lambda_0 \mu_0 {\mathbb P}}{1-\lambda_0 \mu_0}+ \frac{\mu_0 }{ \lambda_0-\mu_0}{\mathbb Q}\right) \ket{g_1}\bra{g_2} =0
\end{array} \end{equation} 
and are clearly equivalent to (\ref{lammu}).

\hfill $\Box$

\medskip

To assure that constraint $\Phi^2={\mathbb I}$ is preserved under Darboux transformation we used relations
 $\bra{q_2} = \bra{p_2} \Phi $ and $\ket{g_1} = \Phi \ket{f_1}$,
 but the quantities involved in this relations,
by virtue of Theorem~\ref{Darb} obeys the following evolution
\begin{equation} \begin{array}{l} \label{evolution}
\bra{p_2} = \bra{p_{20}} \Psi ^{-1} (\lambda_0) \ ,  \quad \ket{f_1} = \Psi  (\mu_0) \ket{f_{10} } \\ 
\bra{q_2} = \bra{q_{20}} \Psi ^{-1} \left( \lambda_0^{-1} \right) \ , \quad 
\ket{g_1} = \Psi  \left( \mu_0^{-1} \right) \ket{g_{10} } \ ,
\end{array} \end{equation} 
where $\Psi  (\lambda)$ is the solution of the Lax pair, and 
$\bra{p_{20}}$, $\bra{q_{20}}$, $\ket{f_{10}}$, $\ket{g_{10}}$ are constant vectors (initial data), so
the questions arises whether these relations are compatible with the evolution? The answer is positive

Due to the reduction (\ref{con-p}) we have 
\begin{equation} 
\Psi  \left( \lambda_0^{-1} \right) = \Phi \Psi  (\lambda_0) {\mathbb S}\ , \quad  \Psi  \left( \mu_0^{-1} \right) = \Phi \Psi  (\mu_0) {\mathbb S}\ ,
\end{equation} 
which means that
\begin{equation} 
\bra{q_2} = \bra{q_{20}} {\mathbb S}\Psi ^{-1} (\lambda_0) \Phi \ ,  \quad \ket{g_1} = \Phi \Psi  (\mu_0) {\mathbb S}\ket{g_{10} } \ ,
\end{equation} 
and, as a conclusion we get that  
to satisfy relations in Lemma~\ref{tw3}  it is enough to set 
 $\bra{q_{20}} = \bra{p_{20}} {\mathbb S}$ and $\ket{g_{10}} = {\mathbb S}\ket{f_{10}}$.

We can now formulate the main conclusion of the article
\begin{corollary} \label{InvolutiveSolutions}
Having a solution $\Phi$ of system (\ref{chep2}) and a corresponding solution $\Psi(\lambda)$ of Lax pair  (\ref{Lp}),
by virtue of Theorem~\ref{Darb} and considerations above we have 
\begin{equation} \begin{array}{l} \label{p2}
\bra{p_2} = \bra{p_{20}} \Psi ^{-1} (\lambda_0) \ ,  \quad \ket{f_1} = \Psi  (\mu_0) \ket{f_{10} } \\ 
\end{array} \end{equation} 
where 
$\bra{p_{20}}$,  $\ket{f_{10}}$, are constant vectors (initial data).
Equations (\ref{greek1}) under the reduction considered in this section with
\begin{equation}  \label{molo}
  \lambda_1 =\lambda_0, \quad  \lambda_2 =\lambda_0^{-1}, \quad \mu_1=\mu_0, \quad \mu_2 =\mu_0^{-1}, \quad \lambda_0,\mu_0\neq 0,1 
\end{equation} 
and after introducing quantities
\begin{equation}  \label{xy}
 x := \frac{\braket{p_2}{f_1}}{\lambda_0 - \mu_0} \ , \qquad y := \frac{\bra{p_2}\Phi \ket{f_1}}{1 - \lambda_0 \mu_0} \ ,
\end{equation} 
take form
\begin{equation} \ket{p_1} = \lambda_0^{-1} \left(\ket{f_1} y^{-1}+ \mu_0^{-1} \Phi \ket{f_1} x^{-1} \right) (xy^{-1}-yx^{-1})^{-1}
\end{equation} 
\begin{equation} \ket{q_1} =- \left(\ket{f_1} x^{-1}+ \mu_0^{-1} \Phi \ket{f_1} y^{-1} \right) (xy^{-1}-yx^{-1})^{-1}. 
\end{equation} 
Then the Darboux matrix that respects the constraint $\Phi^2={\mathbb I}$ is
\[{\mathbb D} (\lambda) = {\mathbb I} +\frac{\lambda_0 \ket{p_1} \bra{p_2}}{\lambda-\lambda_0}+
\frac{\ket{q_1}\bra{p_2}\Phi}{\lambda \lambda_0-1}\]
and family of functions $\widetilde{\Phi}$ given by
\begin{equation} 
\widetilde{\Phi}=\Phi -\ket{p_1} \bra{p_2} \Phi-\ket{q_1}\bra{p_2}
\end{equation} 
is family of solutions of system (\ref{chep2}).
\end{corollary}

\subsection{Solutions}\label{sec32}
We apply the procedure described in Corollary (\ref{InvolutiveSolutions}) to background solution~$\Phi$ of system (\ref{chep2})
\[\Phi =\begin{bmatrix} 0 & a^{m_1} b^{m_2} \\ \frac{1}{a^{m_1}}  \frac{1}{b^{m_2}} & 0 \end{bmatrix}.\]
where $a$ and $b$ are complex constants.
For this solution the Lax pair (\ref{Lp}) takes form
\[
T_1 \Psi (\lambda) = \begin{bmatrix}
\frac{\lambda -a}{\lambda -1} & 0 \\
 0 & \frac{a \lambda -1}{a (\lambda -1)} \\
\end{bmatrix}
\Psi (\lambda), \qquad 
T_2 \Psi (\lambda) = 
\begin{bmatrix}
\frac{b+\lambda }{\lambda +1} & 0 \\
 0 & \frac{b \lambda +1}{b (\lambda +1)} \\
\end{bmatrix}
\Psi(\lambda),
\]
from which we get the following fundamental solution
\[
\Psi(\lambda)=
\begin{bmatrix}
 0 & \left(\frac{\lambda-a}{\lambda-1}\right)^{m_1} \left(\frac{b+\lambda}{\lambda+1}\right)^{m_2} \\
 \left(\frac{1-a \lambda}{a-a \lambda}\right)^{m_1} \left(\frac{b \lambda+1}{b
   \lambda+b}\right)^{m_2} & 0 
\end{bmatrix}.
\]
Taking 
\[ \bra{p_{20}}= [c_2, c_1], \qquad \ket{f_{10} } =\begin{bmatrix}c_4 \\ c_3 \end{bmatrix}, \]
where $c_i$ are arbitrary complex constants we obtain from procedure described in Corollary~(\ref{InvolutiveSolutions})
 solution $\widetilde{\Phi}$ of system (\ref{chep2})
\[\widetilde{\Phi}=\frac{1}{D} \begin{bmatrix} M_{11} & M_{12} \\ M_{21} & -M_{11} \end{bmatrix}, \]
where
\[  D=
   \lambda_0 \mu_0  \left(\gamma + \delta -\alpha - \beta\right) \left(\alpha + \beta+\gamma +\delta \right) \]
\[  M_{12}=
   a^{m_1} b^{m_2} \left(\mu_0 \gamma + \lambda_0 \delta-\alpha -\lambda_0 \mu_0 \beta \right) \left(\alpha +\lambda_0 \mu_0 \beta+\mu_0 \gamma + \lambda_0 \delta \right)\]   
\[  M_{21}=
   a^{-m_1} b^{-m_2} \left(\lambda_0 \gamma + \mu_0 \delta -\lambda_0 \mu_0 \alpha - \beta\right) \left(\lambda_0 \mu_0 \alpha + \beta+\lambda_0 \gamma + \mu_0 \delta \right) \]     
\[
M_{11}=
 \lambda_0 \left(\mu_0^2-1\right) \left(\alpha \gamma - \beta \delta \right) +\mu_0 \left(\lambda_0^2-1\right) \left( \alpha \delta - \beta \gamma \right) .
\] 
   and where
  \[ \alpha := c_1 c_4 (\lambda_0-\mu_0)  \left(a \lambda_0-1\right)^{m_1} \left(a \mu_0-1\right)^{m_1} \left(b
   \lambda_0+1\right)^{m_2} \left(b \mu_0+1\right)^{m_2}\]
 \[ \beta := c_2 c_3 (\lambda_0-\mu_0)  \left(\lambda_0-a\right)^{m_1} \left(\mu_0-a\right)^{m_1} \left(b+\lambda_0\right)^{m_2} \left(b+\mu_0\right)^{m_2}\]
 \[\gamma:= c_1
   c_3 (\lambda_0 \mu_0-1) \left(a \lambda_0-1\right)^{m_1} \left(\mu_0-a\right)^{m_1} \left(b \lambda_0+1\right)^{m_2}
   \left(b+\mu_0\right)^{m_2} \]
\[\delta:= c_2 c_4 (\lambda_0 \mu_0-1) 
   \left(\lambda_0-a\right)^{m_1} \left(a \mu_0-1\right)^{m_1} \left(b+\lambda_0\right)^{m_2} \left(b \mu_0+1\right)^{m_2}. \]
In this particular case ($N=2$) we are able to find the large class of solutions that satisfy also unitary constraint.
\section{Unitary reduction $\Phi^{\dagger} \Phi ={\mathbb I}$ in the case $N=2$}
\label{unitary}
In general unitary constraint $\Phi^{\dagger} \Phi ={\mathbb I}$ is not compatible with equation (\ref{che}). However,
in the particular case $N=2$ set of matrices 
\[{\mathcal S} =\left\{ \begin{bmatrix} a & b +i c \\ b-i c & -a \end{bmatrix} \, | \, a,b,c \in {\mathbb R}, a^2+b^2+c^2=1\right\}\]
is not only unitary, hermitian and involutive at the same time
$\Phi^{-1}=\Phi^{\dagger}=\Phi$, but it follows that
anti-commutator of two elements of $ \mathcal S $ gives
\[
 \begin{array}{l} 
\begin{bmatrix} a_1 & b_1 +i c_1 \\ b_1-i c_1 & -a_1 \end{bmatrix} 
\begin{bmatrix} a_2 & b_2 +i c_2 \\ b_2-i c_2 & -a_2 \end{bmatrix}+ 
\begin{bmatrix} a_2 & b_2 +i c_2 \\ b_2-i c_2 & -a_2 \end{bmatrix} 
\begin{bmatrix} a_1 & b_1 +i c_1 \\ b_1-i c_1 & -a_1 \end{bmatrix} =
\\
=2 (a_1 a_2 + b_1 b_2+ c_1 c_2) \begin{bmatrix} 1 & 0 \\ 0 & 1 \end{bmatrix} 
\end{array}
\]
and therefore if at some point of the lattice $\Phi \in {\mathcal S}$ and at neighboring points $T_1 \Phi\in {\mathcal S}$ and $T_2\Phi\in {\mathcal S}$ then due to (\ref{che}) we have $(T_1T_2 \Phi)^2 ={\mathbb I}$  and moreover $(T_1T_2 \Phi)^{\dagger} (T_1T_2 \Phi)={\mathbb I}$.

Applying the procedure from section \ref{sec32} to seed solution (clearly $ \Phi^{\dagger} \Phi ={\mathbb I}$ and $ \Phi^2 ={\mathbb I}$) 
\[\Phi =\begin{bmatrix} 0 & e^{i a m_1} e^{i b m_2} \\ e^{-i  a m_1}  e^{-i b m_2} & 0 \end{bmatrix}\]
where $a$ and $b$ this time are real parameters,
it is enough to set
\[\begin{bmatrix}c_4 \\ c_3 \end{bmatrix}=\begin{bmatrix}{\overline c_2} \\ {\overline c_1} \end{bmatrix}, \qquad
\mu_0={\overline \lambda_0}\]
to guarantee that $ \tilde{\Phi}^{\dagger} \tilde{\Phi} ={\mathbb I}$ .
Therefore 
{\tiny 
\begin{equation} \tilde{\Phi}= 
\frac{1}{|\lambda_0|^2((s+\bar{s})^2+(r_1+r_2)^2)}
\begin{bmatrix} 
i r_1 [\Lambda s-\bar{\Lambda} \bar{s}]+i r_2 [\bar{\Lambda} s- \Lambda \bar{s}] & 
e^{i m_1 a} e^{i m_2 b} \left( (s+|\lambda_0|^2 \bar{s})^2+({\overline \lambda_0} r_1+\lambda_0 r_2)^2 \right)
\\ 
e^{-i m_1 a} e^{-i m_2 b}\left( (|\lambda_0|^2 s+ \bar{s})^2+(\lambda_0 r_1+ {\overline \lambda_0}r_2)^2 \right) &
-i r_1 [\Lambda s-\bar{\Lambda} \bar{s}]-i r_2 [\bar{\Lambda} s- \Lambda \bar{s}] \end{bmatrix},
\end{equation}
}
where 
\[\Lambda =\lambda_0 ({\overline \lambda_0}^2-1)\]
\[r_1=(|\lambda_0|^2-1) |c_1|^2 |e^{i a} \lambda_0-1|^{2m_1} |e^{i b} \lambda_0+1|^{2m_2} \]
\[r_2=(|\lambda_0|^2-1) |c_2|^2 |e^{i a} {\overline \lambda_0}-1|^{2m_1} |e^{i b} {\overline \lambda_0}+1|^{2m_2} \]
\[s=i({\overline \lambda_0}-\lambda_0) c_1 {\overline c_2} (e^{i a} \lambda_0-1)^{m_1}({\overline \lambda_0}-e^{-i a})^{m_1}
(e^{i b} \lambda_0+1)^{m_2}({\overline \lambda_0}+e^{-i b})^{m_2}
\]
is a family of solution of the system
\begin{equation}
\label{chep2pd}
\Delta_1 ((\Delta_2 \Phi) \Phi )+\Delta_2 ((\Delta_1 \Phi) \Phi )=0,\quad \Phi^2={\mathbb I}, \quad \Phi^{\dagger} \Phi={\mathbb I}.
\end{equation}
i.e. for the discrete  chiral field that takes values on two-dimensional sphere $S^2$.

\section{Conclusions}
We have constructed the Darboux transformation for the difference equation
{\tiny
\begin{equation} \label{EndEq}
\Phi(m_1+1,m_2+1)\left(\Phi^{-1}(m_1+1,m_2)+\Phi^{-1}(m_1,m_2+1)\right) =\left(\Phi (m_1+1,m_2)+\Phi (m_1,m_2+1)\right) \Phi^{-1}(m_1,m_2)
\end{equation}}
and then have reduced the transformation so that it preserves the involutive  constraints
\begin{equation} \label{EndC} \left(\Phi(m_1,m_2)\right)^2=1.\end{equation}
Unlike the continuous case, system (\ref{EndEq}-\ref{EndC}) in general is not compatible with the unitary constraint (\ref{u}). An open problem remains to find an integrable discretisation of the chiral field models (\ref{ContinuousR}) and (\ref{Continuous}) that admits reductions to classical Lie algebras  (or Lie groups) and further involutive reductions.
Moreover, to the best of our knowledge,  integrable discretisation of elliptic chiral field model (i.e. when in (\ref{ContinuousR}) $u=x+iy$ and $v=x-i y$  where 
$x,\, y$ are real independent variables)
has not been discussed in the literature, so far. In particular, integrable discretization of the original (i.e. elliptic) Ernst equation
is not known, yet. However, in the autonomous case one can replace the discrete operator in (\ref{iso-d}) with discrete Schr\"odinger
operator \cite{57S,NieSa,DoNie} or equivalently to apply the so called sublattice approach~\cite{sub45,sub7} to first equation of (\ref{iso-d}) to get its elliptic version.

\section*{Acknowledgements}
\parbox{.135\textwidth}{\begin{tikzpicture}[scale=.03]
\fill[fill={rgb,255:red,0;green,51;blue,153}] (-27,-18) rectangle (27,18);
\pgfmathsetmacro\inr{tan(36)/cos(18)}
\foreach \i in {0,1,...,11} {
\begin{scope}[shift={(30*\i:12)}]
\fill[fill={rgb,255:red,255;green,204;blue,0}] (90:2)
\foreach \x in {0,1,...,4} { -- (90+72*\x:2) -- (126+72*\x:\inr) };
\end{scope}}
\end{tikzpicture}} \parbox{.85\textwidth}
{This research is part of the project No. 2022/45/P/ST1/03998  co-funded by the National Science Centre and the European Union Framework Programme
 for Research and Innovation Horizon 2020 under the Marie Sklodowska-Curie grant agreement No. 945339. For the purpose of Open Access, the author has applied a CC-BY public copyright licence to any Author Accepted Manuscript (AAM) version arising from this submission. MN is beneficiary of the project.
 FWN was supported by the Foreign Expert Program of the Ministry of Sciences
and Technology of China, Grant No. G2023013065L.}

\bibliographystyle{unsrt}
\bibliography{DCh}

\end{document}